\documentclass[12pt]{amsart}
\usepackage{amssymb,amsmath,mathrsfs,enumerate,mathtools,comment}
\usepackage{tikz-cd}
\usepackage[pdftex,bookmarks=true]{hyperref}
\usepackage{adjustbox}
\usepackage{hhline}
\usepackage{enumerate,graphicx}
\newtheorem*{theorem*}{Theorem}
\theoremstyle{plain}

\newtheorem{theorem}{Theorem}[section]

\newtheorem{proposition}[theorem]{Proposition}
\newtheorem{lemma}[theorem]{Lemma}
\newtheorem{corollary}[theorem]{Corollary}
\theoremstyle{definition}

\newtheorem{example}[theorem]{Example}
\newtheorem{definition}[theorem]{Definition}
\newtheorem{remark}[theorem]{Remark}
\numberwithin{equation}{section}

\newcommand{\abs}[1]{\lvert#1\rvert}
\newcommand{\norm}[1]{\lVert#1\rVert}
\newcommand{\bigabs}[1]{\bigl\lvert#1\bigr\rvert}

\usepackage{array}

\usepackage{bbm}
\usepackage{float}

\numberwithin{equation}{section}

\newcommand{\N}{{\mathbb N}}
\newcommand{\R}{{\mathbb R}}

\newcommand{\E}{{\mathbb E}}

\newcommand{\cF}{{\mathcal F}}

\newcommand{\eqdist}{\stackrel{\rm d}{=}}

\newcommand{\bP}{{\mathbb P}}

\newcommand{\cX}{{\mathcal X}}

\newcommand{\one}{\mathbf{1}}

\newcommand{\es}{\operatorname{ess\ sup}}

\title{On Prudence of Risk Measures}

\author[N.~Gao]{Niushan Gao}
\address{Department of Mathematics, Toronto Metropolitan University, 350 Victoria Street, Toronto, Canada M5B 2K3}
\email{niushan@torontomu.ca}

\author[D.~Leung]{Denny H.~Leung}
\address{Department of Mathematics, National University of Singapore, Singapore 117543}
\email{dennyhl@u.nus.edu}

\author[F.~Xanthos]{Foivos Xanthos}
\address{Department of Mathematics, Toronto Metropolitan University, 350 Victoria Street, Toronto, Canada M5B 2K3}
\email{foivos@torontomu.ca}

\thanks{The first and third authors acknowledge   support from NSERC   Discovery Grants.}



\date{\today}

\begin{document}

\maketitle
\begin{abstract}
Prudence is a stability property of risk functionals recently introduced by Wang and Zitikis and subsequently studied by Amarante and Liebrich. In this paper, we first establish general relationships between prudence and other stability properties, showing, in particular, that weak prudence and prudence coincide for a broad class of convex, law-invariant functionals. We then prove that prudence is preserved by cash-additive hulls of star-shaped functionals under a simple asymptotic condition, and by inf-convolutions of convex, cash-additive, law-invariant prudent functionals. Our results   provide general methods for constructing prudent risk measures from existing prudent functionals.
\end{abstract}

\section{Introduction} 
In practice, risk positions are often approximated by positions with simpler structures, under the assumption that the corresponding risk-measure values exhibit suitable approximation properties. Specifically, let $\cX$ be a collection of random variables representing loss positions, and let $\rho:\cX\rightarrow(-\infty,\infty]$ be a risk functional. Let $(X_n)_{n=1}^\infty \subset \cX$ and $X\in\cX$ be such that $(X_n)$ converges to $X$ in some sense. It is then desirable that $\rho(X)$ can be controlled  by a suitable  limit of the sequence $(\rho(X_n))$. Properties of this type are sometimes referred to as stability properties. 

In the recent literature, Wang and Zitikis \cite{WZ21} introduced the notion of prudence  in their study of axiomatic characterizations of  Expected Shortfall.
A functional $\rho:\cX\rightarrow (-\infty,\infty]$ is said to be \emph{prudent} if $$\rho(X)\leq \lim_n\rho(X_n)\quad  \text{whenever } X_n\xrightarrow{a.s.}X\text{ and }\lim_n\rho(X_n)\text{ exists in }[-\infty, \infty].$$
The motivation is that, when   a sequence $(X_n)$ is used to approximate $X$, the quantity $\lim_n\rho(X_n)$ serves as an upper bound for $\rho(X)$. Prudent risk measures can therefore serve useful regulatory purposes. The importance of prudence is further underscored by the fact that it was identified in \cite{WZ21} as one of the four axioms characterizing Expected Shortfall. Recently, Amarante and Liebrich \cite{AL24} studied this property in depth for distortion risk measures and general risk functionals. In their analysis, they also observed connections between prudence and other stability properties. In particular, prudence is  equivalent to the super Fatou property. 

The super Fatou property was originally introduced  by Gao and Munari \cite{GM20} in the study of  surplus-invariant risk measures. More generally, many stability properties can be formulated within a unified framework. Let  $\eta$ denote  a mode of  convergence on $\cX$. A common stability requirement is that
\begin{align}
\label{stability}
    \rho(X)\leq \liminf_n\rho(X_n),\quad \text{whenever }X_n\xrightarrow{\eta}X,
\end{align} 
 In other words, $\rho$ is required to be lower semicontinuous with respect to $\eta$. 
 
When $\eta$ is the convergence induced by a  topology on $\cX$, for example the norm topology, this condition reduces to  the usual notion of lower semicontinuity with respect to that topology. When $\eta$ is dominated almost sure convergence in $\cX$, that is, $X_n\xrightarrow{\eta}X$ means that  $X_n\stackrel{a.s.}{\longrightarrow}X$ and there exists $Y\in \cX$ such that $\abs{X_n}\leq Y$ a.s.\ for all $n\in\N$, it becomes the \emph{Fatou property} of $\rho$, which  was introduced by Delbaen \cite{D02} in the study of dual representations  of risk measures on $L^\infty$   to ensure that the dual elements in the representation can be chosen from $L^1$. 
This property has since been  extensively studied in the literature, see, e.g., Chen et al.~\cite{CGLL22}, Gao et al.~\cite{GLMX18}, Gao et al.~\cite{GLX19}, Gao and Xanthos \cite{GX18}, Jouini et al.~\cite{JST06}, Liebrich and Munari \cite{LM25}. 

When $\eta$ is almost sure  convergence, the resulting property  is called the \emph{super Fatou property}. 
Recall that $\rho$ is said to be \emph{weakly lower semicontinuous} if $\eta$ is convergence in distribution  in \eqref{stability}. As observed by Gao, Munari and Xanthos \cite[Remark 3.7]{GMX20}, weak lower semicontinuity is equivalent to the super Fatou property for law-invariant functionals, owing to the  classical Skorohod representation theorem.   Hence, for law-invariant functionals, prudence, the super Fatou property, and weak lower semicontinuity are equivalent notions. This equivalence highlights the fundamental role of prudence among stability properties of risk functionals   and naturally raises questions regarding its characterization and preservation.

Beyond the fact that Expected Shortfall is prudent, \cite{GMX20} raised the more general question of whether Haezendonck--Goovaerts risk measures are prudent. This question was answered affirmatively by Amarante and Liebrich \cite{AL24}. Nevertheless, despite its close connections with several classical stability properties, relatively little is known about the structural properties of prudence and about operations that  generate or preserve prudent risk measures.

The present paper contributes to this line of research from several perspectives.
Section 2 investigates the relationship between  prudence and other stability properties. In particular, we establish the equivalence between weak prudence and prudence for convex, law-invariant functionals under a natural and mild condition. This result provides new insight into a question raised by Amarante and Liebrich \cite{AL24} concerning the relationship between the two notions. 

Section~3 studies the preservation of prudence under the cash-additive hull operation and establishes a broad and easily verifiable sufficient condition. The cash-additive hull is a standard tool for constructing cash-additive risk measures from more general functionals, and notable examples obtained in this way include Expected Shortfall and Haezendonck--Goovaerts risk measures. Consequently, our result provides a convenient method for constructing prudent cash-additive risk measures from prudent star-shaped functionals.

Finally, Section~4 proves that prudence is preserved by inf-convolutions of convex, cash-additive, law-invariant functionals. Since inf-convolution is a fundamental operation in capital allocation and   risk sharing, this result shows that prudence is stable under a natural aggregation mechanism and identifies a broad class of prudent functionals that remain prudent under aggregation.

\subsection{Notation and Preliminary Facts}
Throughout the paper, $(\Omega,\mathcal{F},\mathbb{P})$ stands for a non-atomic probability space. 
Let $L^0:=L^0(\Omega)$ be the space of all random variables on $(\Omega,\mathcal{F},\mathbb{P})$.
A vector subspace  $\cX$ of $L^0$ is called a \emph{lattice ideal} if $Y\in\cX$ whenever $X\in \cX$ and $\abs{Y}\leq \abs{X}$. A lattice ideal $\cX$ is called a \emph{Banach function space} if it is equipped with a complete norm satisfying that $\norm{Y}\leq\norm{X} $ whenever $\abs{Y}\leq \abs{X}$ in $\cX$. A Banach function space $\cX$ is called a {\em rearrangement-invariant (r.i.) space} if  $X\in \cX$ and $Y \eqdist X$ imply   $Y\in \cX$ and $\|Y\|=\|X\|$. Here, $Y\eqdist X$  means that the two random variables $X$ and $Y$ have the same distribution.  It is well-known that every r.i.\ space $\cX$ over a non-atomic probability space satisfies $$L^\infty\subset \cX\subset L^1.$$

\smallskip

Let $\cX$ be a lattice ideal of $L^0$ containing $L^\infty$. Let $\mathbf{1}$ be the constant one function. A proper functional $\rho:\cX\rightarrow(-\infty,\infty]$ is called a \emph{coherent risk measure} if it is
\begin{enumerate}
    \item[(i)] monotone, i.e., $\rho(X_1)\leq \rho(X_2)$ whenever $X_1,X_2\in\cX$ satisfies $X_1\leq X_2$,
    \item[(ii)] subadditive, i.e., $\rho(X_1+X_2)\leq \rho(X_1)+\rho(X_2)$ for any $X_1,X_2\in\cX$,
    \item[(iii)] positive homogeneous, i.e., $\rho(\lambda X)=\lambda \rho(X)$ for any $X\in\cX$ and   $\lambda\geq 0$,
    \item[(iv)] cash additive, i.e., $\rho(X+m \mathbf{1})=\rho(X)+m$ for any $X\in\cX$ and $m\in\R$.
\end{enumerate}
A functional $\rho:\cX\rightarrow(-\infty,\infty]$ is said to be \emph{convex} if $\rho\big(\lambda X_1+(1-\lambda)X_2\big)\leq \lambda \rho(X_1)+(1-\lambda)\rho(X_2)$ for any $X_1,X_2\in\cX$ and $\lambda\in[0,1]$. $\rho$ is said to be \emph{law invariant} if $\rho(X_1)=\rho(X_2)$ for any $X_1,X_2\in\cX$ with $X_1 \eqdist X_2$.

Finally, recall that a functional $\rho:\cX\rightarrow (-\infty,\infty]$ has the super Fatou property if $\rho(X)\leq \liminf_n\rho(X_n)$ whenever $X_n\xrightarrow{a.s.}X$. As observed in \cite{AL24}, prudence is equivalent to the super Fatou property. Accordingly, many arguments in this paper establish prudence by verifying the super Fatou property.

 \section{Stability properties of risk functionals}

Let $\cX$ be a lattice ideal of $L^0$.
Recall  from \cite{AL24} that a functional $\rho:\cX\rightarrow (-\infty,\infty]$ is said to be \emph{weakly prudent} if $\rho(X)\leq \lim_n\rho(X_n)$ whenever $X_n\xrightarrow{a.s.}X$ and $\lim_n\rho(X_n)$ exists in $\mathbb{R}$. Prudence requires the same inequality whenever
$\lim_n\rho(X_n)$ exists in $[-\infty,\infty]$. Thus, prudence is a priori stronger than weak prudence. A natural question is whether the two notions are  equivalent under additional   assumptions. In particular, Amarante and Liebrich \cite[Sect.~3.2]{AL24} asked  whether weak prudence and prudence are  equivalent for law-invariant functionals. 
The main result of this section  provides a broad sufficient condition for an affirmative answer.

To state the result, we introduce a mild boundedness condition.  A functional $\rho:\cX\rightarrow (-\infty,\infty]$ is said to be \emph{order bounded below} if $$\inf\big\{\rho(X):X\in[U,V]\big\}>-\infty$$ for every order interval $[U,V]:=\{X\in\cX:U\leq X\leq V\}$, where $U,V\in\cX$ satisfy $U\leq V$. Order-boundedness-type conditions have recently emerged as useful substitutes for monotonicity; see Gao and Xanthos \cite{GX25} for a discussion.

\begin{lemma}\label{(P) to (F)}
Let $\cX$ be a lattice ideal of $L^0$.
Assume that $\rho:\cX\to (-\infty,\infty]$ is  weakly prudent and order bounded below.  Then $\rho$ has the Fatou property.
\end{lemma}

\begin{proof}
Let $(X_n)\subset \cX$ and $X\in\cX$ be such that $X_n\xrightarrow{a.s.}X$ and  $\abs{X_n}\leq Y$ for some $Y\in \cX$ and all $n\geq 1$. By passing to a subsequence, we may assume that $\lim_n\rho(X_n)=\liminf_n\rho(X_n)\in[-\infty,\infty]$. Since  $(X_n)\subset [-Y,Y]$ and $\rho$ is order bounded below,  we have $\inf_n\rho(X_n)>-\infty$. Therefore,   $\lim_n\rho(X_n)\in(-\infty,\infty]$. If $\lim_n\rho(X_n)=\infty$, then trivially $\rho(X)\leq \lim_n\rho(X_n)$. Otherwise, $\lim_n\rho(X_n)\in\mathbb{R}$, and  weak prudence yields $\rho(X)\leq \lim_n\rho(X_n)=\liminf_n\rho(X_n)$. Hence, $\rho$ has the Fatou property.
\end{proof}

We now present the main result of this  section. It is applicable not only to measures of risk but also to measures of variability (see, e.g., Bellini et al.~\cite{BFWW22}).

\begin{theorem}\label{wekp=p}
Suppose that $\mathcal{X}$ is an r.i.\ space and let $\rho:\mathcal{X}\to \mathbb{R}$ be a law-invariant, weakly prudent functional. Assume that $\rho$ is convex and order bounded below.  Then $\rho$ is prudent. 
\end{theorem}

\begin{proof}
Define  $f:\R\to\R$ by $f(t)=\rho(t\mathbf{1})$ for $t\in\R$. Suppose first that $f$ is a constant function, say, $f\equiv C$. Observe that $\rho$ is convex, law invariant and has the Fatou property by Lemma \ref{(P) to (F)}. By \cite[Theorem 3.1]{ML20}, every proper convex law-invariant functional with the Fatou property is $\sigma(\cX,L^\infty)$-lower semicontinuous.   Hence, by \cite[Proposition 5.6]{BKMS21}, 
$$ \rho(X)\geq \rho(\E[X]\mathbf{1})=f(\E[X])=C,\qquad \forall X\in\cX.$$
Thus, $\rho$ is bounded below by $C$ on $\cX$. The desired prudence now follows immediately from   weak prudence.

Next, assume that $f$ is non-constant. 
Let $(X_n)\subset \cX$ and $X\in\cX$  be such that $X_n\to X$ a.s. Assume, for contradiction, that $\rho(X)> \liminf_n\rho(X_n)$.   By passing to a subsequence, we may choose $a \in \R$ such that $$\rho(X_n) < a < \rho(X),\qquad \forall n\geq 1.$$
By Egoroff's theorem, passing to a further subsequence if necessary, we may assume   that there is a sequence $(\Omega_n)$ of measurable sets such that 
\begin{align}\label{temf4}\Omega_n \uparrow \Omega, \qquad \bP(\Omega_n)<1,\qquad  |X_n-X|\mathbf{1}_{\Omega_n} \leq \frac{\mathbf{1}_{\Omega_n}}{n},\qquad \forall n\geq 1.\end{align}

For each $n\geq 1$, define $f_n:\R\to \R$ by
\[  f_n(t) = \rho(X_n\mathbf{1}_{\Omega_n}+ t\mathbf{1}_{\Omega_n^c}),\qquad\forall t\in\R.\]
On the one hand,    by \cite[Theorem 3.1]{ML20} and \cite[Proposition 5.6]{BKMS21}, 
\begin{align}\label{temf2}f_n(c_n) \leq \rho(X_n)<a,\quad \text{where } c_n = \E[X_n\mathbf{1}_{\Omega_n^c}]/\bP(\Omega_n^c).\end{align}
On the other hand,  by \cite[Theorem 3.1]{ML20} and \cite[Proposition 5.6]{BKMS21}  again,
\begin{align}\label{temf1}
   f_n(t) \geq \rho(r_t\mathbf{1})=f(r_t),\quad \text{where }r_t = \E[X_n\mathbf{1}_{\Omega_n}+ t\mathbf{1}_{\Omega_n^c}].
\end{align}
Since $\rho$ is convex, $f$ is  convex   on $\R$. As $f$ is non-constant, either $\lim_{t\to \infty}f(t)=\infty$ or $\lim_{t\to-\infty} f(t)=\infty$. Therefore,  by \eqref{temf1}, there exists $d_n\in\R$ such that 
\begin{align}\label{temf3}f_n(d_n)>a.\end{align}
Since $f_n:\R\to\R$ is convex and real-valued, it is  continuous on $\R$. Combining \eqref{temf2} and \eqref{temf3}, we obtain $t_n\in\R$ such that $$f_n(t_n)=a.$$
Set $Y_n = X_n\mathbf{1}_{\Omega_n}+ t_n\mathbf{1}_{\Omega_n^c}$.
Then $\rho(Y_n) = f_n(t_n) =a$ for all $n\geq 1$ and  $Y_n\to X$ a.s.\ by \eqref{temf4}.
Since $\rho$ is weakly prudent, $\rho(X) \leq \lim_n\rho(Y_n) =a$, contradicting the choice of $a$.
\end{proof}

\begin{remark}
The conclusion of Theorem \ref{wekp=p} remains valid if $\rho$ is allowed to  take values in $(-\infty,\infty]$ and satisfies  $$\rho(X_n)\to \rho(X)\qquad \text{whenever}\qquad \|X_n-X\|_\infty \to 0.$$
The same proof applies, since this continuity assumption ensures the existence of the intermediate values $t_n$ used in the argument.
\end{remark}

We next study when the Fatou property implies  prudence (equivalently, the super Fatou property) under   additional assumptions. One  sufficient condition is formulated in terms of monotonicity  and solidity. For convenience, we collect the relevant notions below.
\begin{definition}\label{solidmono}
Let $\cX$ be a lattice ideal of $L^0$. We say that a functional $\rho:\cX \to (-\infty,\infty]$ is 
\begin{enumerate}
\item[(M)] monotone if $\rho(X) \leq \rho(Y)$ whenever $X \leq Y$;
\item[(M+)] positive monotone if $\rho(X) \leq \rho(Y)$ whenever $0 \leq X \leq Y$;
\item[(S)] solid if $\rho(X) = \rho(|X|)$ for all $X \in \cX$;
\item[(S+)] positive solid if $\rho(X) = \rho(X^+)$ for all $X \in \cX$.
\end{enumerate}
\end{definition}

We note that positive solidity was equivalently introduced in \cite{GM20} under the name of surplus invariance for the functional $\widetilde{\rho}(X)=\rho(-X)$.

\begin{lemma}\label{solidity}
For a functional $\rho:\cX\to(-\infty,\infty]$, the following statements hold.
\begin{enumerate}
\item $\rho$ satisfies {\rm (M+)} and {\rm (S)} if and only if $ \rho(X)\leq \rho(Y)$ whenever  $\abs{X}\leq \abs{Y}$. 
\item $\rho$ satisfies  {\rm (M+)} and  {\rm (S+)} if and only if $\rho(X)\leq \rho(Y)$ whenever $X^+\leq Y^+$.
\end{enumerate}
\end{lemma}

\begin{proof}
(1) The ``only if'' direction is immediate. For the converse, suppose that
 $\rho(X)\leq \rho(Y)$ whenever $\abs{X}\leq \abs{Y}$. Then (M+) follows immediately.  To verify  (S), fix $X\in\cX$ and take $Y=|X|$. Since $\abs{X}\leq \bigabs{\abs{X}}$, we have 
$\rho(X)\leq\rho(\abs{X})$. Similarly, since $\bigabs{\abs{X}}\leq \abs{X}$, we obtain
$\rho(\abs{X})\leq \rho(X)$. Therefore,  $\rho(X)=\rho(\abs{X})$, and $\rho$ satisfies (S). This proves (1). 

The proof of (2) is analogous.
\end{proof}

\begin{proposition}\label{solid+fatou}
Suppose that $\rho:\cX\to(-\infty,\infty]$ has the Fatou property and satisfies {\rm (M+)}  and either {\rm (S)} or {\rm (S+)}. Then $\rho$ is prudent.
\end{proposition}

\begin{proof}
 Suppose that $\rho$ satisfies (M+) and  (S). Let $(X_n)\subset \cX$ and $X\in \cX$ be such that $X_n\xrightarrow{a.s.}X$. Then $$0\leq \inf_{m\geq n}\abs{X_m}\uparrow \abs{X}.$$
Thus $$\rho(X)=\rho(\abs{X})\leq \liminf_n\rho\left(\inf_{m\geq n}\abs{X_m}\right)\leq \liminf_n \rho(\abs{X_n})=\liminf_n\rho(X_n),$$
where the first and last equalities follow from (S) and the second inequality follows from (M+).
This proves that $\rho$ has the super Fatou property, and hence is prudent.

The case of  (M+) and (S+) can be proved similarly.
\end{proof}

\begin{remark}
Note that {\rm (M)} implies {\rm (M+)}. Hence, Proposition~\ref{solid+fatou} extends \cite[Theorem~21]{GM20}, which establishes the super Fatou property (hence prudence) under {\rm (M)} and {\rm (S+)}. In addition to replacing {\rm (M)} by the weaker condition {\rm (M+)}, Proposition~\ref{solid+fatou} also covers the case of {\rm (M+)} and {\rm (S)}.
\end{remark}

Finally, we derive several prudence-type consequences of the Fatou property. Although the Fatou property alone is generally insufficient to imply prudence, it still yields super-Fatou-type inequalities for sequences that are bounded below.

\begin{proposition}\label{p2.1}
Let $\cX$ be a lattice ideal of $L^0$ and let $\rho:\cX\to (-\infty,\infty]$ be a proper monotone functional that satisfies the Fatou property at some $X\in \cX$.
If $(X_n)\subseteq \cX$ converges to $X\in \cX$ a.s.\ and $(X_n)$ is bounded below in $\cX$,  then $\rho(X) \leq \liminf_n\rho(X_n)$.
\end{proposition}

\begin{proof}
Let $Y\in \cX$ be such that  $X_n\geq Y$ for all $n\geq 1$.  Replacing $Y$ by  $X\wedge Y$ if necessary, we may assume that $Y\leq X$.
Since 
\[ Y \leq \inf_{m\geq n}X_m \uparrow X \text{ as $n\to \infty$},\]
the sequence $\big(\inf_{m\geq n}X_m\big )_n$ converges to $X$ under dominated almost sure convergence.  By the assumption, $\rho$ has the Fatou property  at $X$ and $\rho$ is monotone.  Thus
\[ \rho(X) \leq \liminf_n\rho\Big(\inf_{m\geq n}X_m\Big) \leq\liminf_n \rho(X_n).
\]
\end{proof}

The lower-bound assumption in Proposition~\ref{p2.1} can be weakened from boundedness below in $\cX$ to boundedness below in $L^1$ when the functional is additionally quasiconvex and law invariant. Recall that every r.i.\ space $\cX$ satisfies $L^\infty\subset \cX\subset L^1$.

\begin{corollary}
Let $\cX$ be an r.i.\ space and let $\rho:\cX\to (-\infty,\infty]$ be a proper, monotone, quasiconvex  and law-invariant  functional that satisfies the Fatou property.
If $(X_n)\subseteq \cX$ converges to $X\in \cX$ a.s.\ and $(X_n)$ is bounded below in $L^1$,  then $\rho(X) \leq \liminf_n\rho(X_n)$.       
\end{corollary}

\begin{proof}
 Recall again from \cite[Theorem 3.1]{ML20} that a proper quasiconvex law-invariant functional with the Fatou property is $\sigma(\cX,L^\infty)$-lower semicontinuous.   Thus, by \cite[Proposition 5.3]{BKMS21}, $\rho$ extends to   a proper, increasing, quasiconvex  and law-invariant  functional $\widetilde{\rho}$ on $L^1$ that is $\sigma(L^1,L^\infty)$-lower semicontinuous, which is equivalent to the Fatou property on $L^1$.
The conclusion now follows from Proposition~\ref{p2.1} applied to $\widetilde{\rho}$ on $L^1$.
\end{proof}

\section{Prudence of Cash-Additive Hulls}
In this section, $\cX$ denotes a vector space  over $\mathbb{R}$, and $\Pi \in \cX\setminus\{0\}$ is fixed. The vector $\Pi$ plays the role of a num\'eraire.

Star-shaped risk measures have attracted considerable attention in recent years; see, for example, Castagnoli et al.~\cite{CCMTW22}, Laeven et al.~\cite{LGZ23, LGZ24},  Liebrich~\cite{L24},  Righi~\cite{R24}.  A functional $\rho:\cX \rightarrow [-\infty,\infty]$  is said to be \emph{star-shaped} if $\rho(0) \in \mathbb{R}$ and 
\begin{align}\label{star}
  \rho(\lambda X) \leq \lambda \rho(X)+(1-\lambda)\rho(0),\quad \text{for any } X\in\cX \text{ and }\lambda \in [0,1].  
\end{align} 
Every convex functional and every  positive homogeneous functional is star-shaped.

Note that if $\rho$ is star-shaped, then the function $t \mapsto \frac{\rho(t \Pi)-\rho(0)}{t}$ is increasing on $(0,\infty)$.  Indeed, let $t_1 >t_2 > 0$. By \eqref{star}, $$\rho(t_2 \Pi)=\rho\Big(\frac{t_2}{t_1}\cdot t_1\Pi\Big) \leq \frac{t_2}{t_1} \rho(t_1\Pi)+\Big(1-\frac{t_2}{t_1}\Big) \rho(0).$$
Rearranging this inequality yields 
$$\frac{\rho(t_2 \Pi)-\rho(0)}{t_2}\leq \frac{\rho(t_1 \Pi)-\rho(0)}{t_1}.$$
By a similar argument, it is also increasing on $(-\infty,0)$. For a star-shaped functional $\rho$, we define 
$$u_{\rho}:=\lim_{t\rightarrow\infty}\frac{\rho(t \Pi)-\rho(0)}{t} \text { and  } l_{\rho}:=\lim_{t\rightarrow-\infty}\frac{\rho(t \Pi)-\rho(0)}{t}.$$

A functional $\rho:\cX\rightarrow [-\infty,\infty]$ is said to be \emph{$\Pi$-additive} if $\rho(X+m\Pi)=\rho(X)+m$ for all $m \in \mathbb{R}$. When $\cX$ is a function space and $\Pi$ is the constant one function, $\Pi$-additivity reduces to the usual cash additivity. For a general functional $\rho:\cX\rightarrow [-\infty,\infty]$, its  \emph{$\Pi$-additive  hull}  was studied in \cite{FK07} and is defined as follows 
$$\rho_{\Pi}(X)=\inf\big\{m +\rho(X-m\Pi): m \in \mathbb{R}\big\},\qquad X\in\cX.$$
It is known that $\rho_{\Pi}$ is the largest $\Pi$-additive functional that is majorized by $\rho$.
The Expected Shortfall, Haezendonck--Goovaerts risk measures, and higher moment coherent risk measures (\cite{K07,DPR10}) are prominent examples of risk measures that arise as cash-additive hulls of suitable functionals. Another example is the generalized optimized certainty equivalent, introduced by Wang, Mao, and Hu \cite{WMH24}.

We now turn to the study of stability properties preserved by $\Pi$-additive hulls. Recall first that,  for a $\Pi$-additive risk functional, $\rho(X)$ is interpreted as the amount of the asset $\Pi$ that can be withdrawn from the position $X$. Consequently, when a $\Pi$-additive hull is used as a risk measure, it is desirable to ensure that it does not attain the value $-\infty$. In what follows, we investigate conditions under which $\Pi$-additive hulls avoid  the value $-\infty$ and preserve  stability properties.

The stability properties considered below include Fatou-type,
super Fatou-type, and topological continuity properties.
To treat these situations simultaneously, we introduce the notion
of a convergence structure. A convergence structure $\eta$ on a vector space $\cX$ is a relation between nets and vectors satisfying the following properties:
\begin{enumerate}
\item[(i)] If $X_{\alpha} \xrightarrow{\eta} X$ and $X_{\alpha} \xrightarrow{\eta} Y$ then $X=Y$.
\item[(ii)] If $X_{\alpha}=X$ for all $\alpha$, then $X_{\alpha} \xrightarrow{{\eta}} X$,
\item[(iii)] If $X_{\alpha} \xrightarrow{{\eta}} X$, then every subnet of $(X_{\alpha})$ converges to $X$ with respect to $\eta$,
\item[(iv)] Suppose that $X_{\alpha} \xrightarrow{\eta} X$ and $Y_{\alpha} \xrightarrow{\eta} Y$. Let $(Z_{\alpha})$ be a net in $\cX$ such that $Z_{\alpha} \in \{X_{\alpha},Y_{\alpha}\}$ for every $\alpha$. Then $Z_{\alpha} \xrightarrow{{\eta}} X$.
\item[(v)]  The vector space operations are continuous with respect to $\eta$.
\end{enumerate}
The pair $(\cX,\eta)$ is called  a \emph{convergence vector space}.  
We refer the reader to \cite{BTW23} for a recent discussion of convergence structures.

Let $(\cX,\eta)$ be a  convergence vector space. We say that a functional $\rho:\mathcal X\to(-\infty,\infty]$ is  \emph{lower semicontinuous} at $X \in \mathcal{X}$ if  $\rho(X) \leq \liminf \rho(X_{\alpha})$ whenever $X_{\alpha} \xrightarrow{{\eta}} X$ in $\cX$.   We say that $\rho$ is \emph{$\sigma$-lower semicontinuous} at $X \in \mathcal{X}$ if the same condition holds for sequences instead of nets.
The notions of continuity and $\sigma$-continuity are defined analogously.

The following lemma is a key ingredient in the proof of the main result of this section.

\begin{lemma}\label{atinfty}
Let $(\cX,\eta)$ be a convergence vector space, $\Pi\in\cX$, and let $\rho:\cX \rightarrow (-\infty,\infty]$ be  a star-shaped  functional with $l_{\rho}<1<u_{\rho}$. Assume that $\rho$ is lower semicontinuous at $t\Pi$ for every $t\in\R$ with $|t|>t_0$ for some $t_0>0$. If $X_{\alpha} \xrightarrow{{\eta}} X \in \cX$ and  either $m_{\alpha} \rightarrow \infty$ or $m_{\alpha} \rightarrow -\infty$, then 
$$\lim_{\alpha}\Big(m_{\alpha}+\rho(X_{\alpha}-m_{\alpha} \Pi)\Big)=\infty.$$
If $\rho$ is $\sigma$-lower semicontinuous at $t\Pi$ for every $t$ with $|t|>t_0$ for some $t_0>0$, then the same conclusion holds  when nets are replaced by sequences.
\end{lemma}

\begin{proof}
Suppose that either $m_{\alpha} \rightarrow \infty$ or $m_{\alpha} \rightarrow -\infty$, and that $X_{\alpha} \xrightarrow{\eta} X$. Assume, for contradiction, that $\lim_{\alpha}\bigl(m_{\alpha}+\rho(X_{\alpha}-m_{\alpha} \Pi)\bigr)\neq \infty$.  Then $$x:=\liminf_{\alpha}\Big(m_{\alpha}+\rho(X_{\alpha}-m_{\alpha} \Pi)\Big) \in \mathbb{R}\cup\{-\infty\}.$$
By passing to a subnet, we may assume that $m_{\alpha}+\rho(X_{\alpha}-m_{\alpha} \Pi) \rightarrow x$. Hence, there exist $c \in \R$ and $\alpha_0 $ such that $$m_\alpha+\rho(X_{\alpha}-m_{\alpha} \Pi)\leq c,\quad\forall \alpha \geq \alpha_0.$$

Fix $t>t_0$.  Set $\delta=-1$ if $m_{\alpha} \rightarrow \infty$, and  $\delta=1$ if $m_{\alpha} \rightarrow -\infty$. By passing to a further subnet, we may assume that  $\delta\cdot (c- m_{\alpha})>t$ for all $\alpha \geq \alpha_0$. Then 
\begin{align*}
\frac{\rho(X_{\alpha}-m_{\alpha}\Pi)}{\delta\cdot (c-m_{\alpha})}\leq \delta,\qquad \forall \alpha\geq \alpha_0.
\end{align*}
Since $\rho$ is star-shaped,  it follows that
\begin{align}\label{basict2}
& \rho\left(t\cdot \frac{X_{\alpha}-m_{\alpha}\Pi}{\delta \cdot (c-m_{\alpha})}\right)\\
\nonumber
\leq&  \frac{t}{\delta\cdot (c-m_{\alpha})}\rho(X_{\alpha}-m_{\alpha}\Pi)
+ \left(1-\frac{t}{\delta\cdot (c-m_{\alpha})}\right)\rho(0) \\
\nonumber\leq &t\delta + \left(1 - \frac{t}{\delta \cdot (c-m_{\alpha})}\right)\rho(0).
\end{align}
Observe that, since the vector operations are continuous,  $$\frac{t}{\delta} \cdot \frac{X_{\alpha}-m_{\alpha}\Pi}{c-m_{\alpha}}=\frac{t}{\delta}\cdot \Big(\frac{X_{\alpha}-c\Pi}{c-m_{\alpha}}+\Pi\Big)\xrightarrow{\eta} \frac{t}{\delta}\Pi.$$
Therefore, by the lower semicontinuity of $\rho$ and \eqref{basict2}, we have 
$$\rho\Big(\frac{t}{\delta}\Pi\Big)\leq \liminf_{\alpha} \rho\left(t \cdot \frac{X_{\alpha}-m_{\alpha} \Pi}{c-m_{\alpha}}\right)\leq t\delta+\rho(0).$$
If $\delta=1$, it follows that $u_{\rho} \leq 1$, contradicting the assumption that $u_\rho>1$. If $\delta=-1$, it follows that   $l_{\rho} \geq 1$, contradicting the assumption that $l_\rho<1$. The proof in the 
$\sigma$-lower semicontinuous case  is identical, with nets replaced by sequences. 
\end{proof}

We are now ready to prove the main result of this section.

\begin{theorem}\label{valuesofch}
Let $(\cX,\eta)$ be a convergence vector space, let $\Pi\in\cX$, and let $\rho:\cX \rightarrow (-\infty,\infty]$ be a proper, star-shaped,  $\sigma$-lower semicontinuous functional satisfying   $l_{\rho}<1<u_{\rho} $. Then  the following hold:
\begin{enumerate}
\item  $\rho_{\Pi}$ takes values in $(-\infty,\infty]$, and the infimum in  the definition of $\rho_{\Pi}(X)$ is attained for every $X\in\cX$.
\item  $\rho_{\Pi}$ is  $\sigma$-lower semicontinuous.
\end{enumerate}
Moreover, if $\rho$ is lower semicontinuous (respectively, $\sigma$-continuous, or continuous), then $\rho_{\Pi}$ is also   lower semicontinuous (respectively, $\sigma$-continuous, or continuous).
\end{theorem}

\begin{proof}
(1) Fix $X\in\cX$. If $\rho_\Pi(X)=\infty$, then the conclusion is immediate. Assume therefore that $\rho_\Pi(X)<\infty$. Take a sequence $(m_n)$ of real numbers such that $$m_n+\rho(X-m_n\Pi)\rightarrow \rho_\Pi(X).$$

Since the  constant sequence $(X)_{n=1}^\infty $ is $\eta$-convergent,  Lemma \ref{atinfty} implies that $(m_n)$ cannot diverge to $\pm \infty$. By passing to a subsequence, we may therefore assume that $m_n\rightarrow m_0$ for some $m_0\in\R$.
Since $X-m_n\Pi\xrightarrow{\eta }X-m_0\Pi$ in $\cX$,  $\sigma$-lower semicontinuity of $\rho$ implies that
$$\rho_\Pi(X)\leq m_0+\rho(X-m_0\Pi)\leq\liminf_n \Big( m_n+\rho(X-m_n\Pi) \Big)=\rho_\Pi(X).$$
Hence, $$\rho_\Pi(X)=m_0+\rho(X-m_0\Pi)\in(-\infty,\infty]$$ 
and the infimum in the definition of $\rho_\Pi(X)$ is attained at $m_0$.

(2) Let $(X_{n})$ be a sequence in $\cX$ such that $X_{n}\xrightarrow{\eta} X \in \cX$. We will show that $\rho_{\Pi}(X) \leq \liminf_{n}\rho_{\Pi}(X_n)$.    By part (1), for any $n\in\N$, there exists $m_{n} \in \mathbb{R}$ such that $$\rho_{\Pi}(X_{n})=m_{n}+\rho(X_{n}-m_{n}\Pi).$$ 

If $(m_{n})$ is unbounded, then, by passing to a subsequence, we may assume that $m_{n} \rightarrow \infty$ or $m_{n} \rightarrow -\infty$.  Lemma \ref{atinfty} then yields $\liminf_n \rho_{\Pi}(X_{n})=\infty$ and the desired inequality follows immediately.  Suppose next that $ (m_{n})$ is bounded. By passing to a subsequence, we may assume that $m_{n}\rightarrow v_0 $ for some $v_0\in \mathbb{R}$. Then $X_{n}-m_{n}\Pi\xrightarrow{{\eta}} X-v_0\Pi$. The $\sigma$-lower semicontinuity of $\rho$ implies 
\begin{align}\label{lowsemi}
    \rho_{\Pi}(X)\leq v_0+\rho(X-v_0\Pi)\leq \liminf_n\Big( m_{n}+\rho(X_{n}-m_{n}\Pi)\Big)=\liminf_n\rho_{\Pi}(X_{n}).
\end{align}
This proves (2).

Finally, we prove the last assertion. The case of lower semicontinuity can be proved in the same way as above. Assume now that $\rho$ is $\sigma$-continuous with respect to $\eta$, and let $X_n\xrightarrow{{\eta}} X$. By part~(1), there exists  $m_0\in\R$ such that $\rho_\Pi(X)=m_0+\rho(X-m_0\Pi)$. Since $X_n-m_0\Pi\xrightarrow{{\eta}} X-m_0\Pi$, the $\sigma$-continuity of $\rho$ yields $$\rho(X_n-m_0\Pi)\rightarrow \rho(X-m_0\Pi).$$
Consequently,
$$\limsup_n\rho_\Pi(X_n)\leq \limsup_n \Big(m_0+\rho(X_n-m_0\Pi)\Big)=m_0+\rho(X-m_0\Pi)=\rho_\Pi(X).$$
Combining this with \eqref{lowsemi}, we obtain $\rho_\Pi(X)=\lim_n\rho_\Pi(X_n)$. 
This proves that $\rho_\Pi$ is $\sigma$-continuous. 
The case of continuity is analogous.
\end{proof}

The following result shows that the condition
$l_\rho<1<u_\rho$ cannot be omitted. Indeed, if this condition fails,
the $\Pi$-additive hull $\rho_\Pi$ may attain the value $-\infty$.

\begin{lemma}\label{badlemma}
Let $\rho:\cX\rightarrow (-\infty,\infty]$ be a star-shaped functional such that $l_\rho>1$ or $u_\rho<1$. Then $\rho_{\Pi}(\Pi)=-\infty$.
\end{lemma}

\begin{proof}
Assume first that $l_\rho>1$. Suppose, for contradiction, that $\rho_{\Pi}(\Pi)>-\infty$. Then there exists $c\in\R$ such that  $$m+\rho(\Pi-m\Pi) \geq c,\qquad\forall m \in \mathbb{R}.$$
Hence, $$\frac{\rho((1-m)\Pi)}{1-m} \leq \frac{c-m}{1-m},\qquad\forall m>1.$$
Letting $m \rightarrow \infty$, we obtain $l_{\rho} \leq 1$, contradicting the assumption that $l_\rho>1$.

The case  $u_{\rho}<1$ is analogous.
\end{proof}

Theorem~\ref{valuesofch} has a wide range of applications. We illustrate some of them below.
\begin{example}\label{conexlowersemi}Let $(\cX,\eta)$ be a convergence vector space, $\Pi\in\cX$, and $\rho:\cX\rightarrow (-\infty,\infty]$ be a proper, star-shaped functional  satisfying   $l_{\rho}<1<u_{\rho} $. 
    \begin{enumerate}
        \item Suppose that $\cX$ is a normed space and $\eta$ is norm convergence. Then Theorem~\ref{valuesofch} implies that  $\Pi$-additive hulls preserve norm continuity and norm lower semicontinuity.
        \item Suppose that $\cX$ is a function space and $\eta$ is dominated a.s.\ convergence. By Theorem~\ref{valuesofch}, if $\rho$ has the Fatou property, then $\rho_\Pi$ has the Fatou property.  Moreover, the theorem also implies that if $\rho$ has the Lebesgue property, then $\rho_\Pi$ has the Lebesgue property. Recall that a functional $\psi$ is said to have the Lebesgue property if $\psi(X)=\lim_n\psi(X_n)$ whenever $X_n\xrightarrow{a.s.} X$ and $(X_n)$ is dominated in $\cX$.
        \item  Suppose that $\cX$ is a function space and $\eta$ is a.s.\ convergence. By Theorem~\ref{valuesofch}, if $\rho$ is prudent (equivalently, $\rho$ has the super Fatou property), then $\rho_\Pi$ is also prudent.
    \end{enumerate}
\end{example}

We now discuss a prominent example related to Example~\ref{conexlowersemi}(3), which was the original motivation for this section. 

We begin by recalling the definitions of Orlicz functions and Orlicz spaces.
A non-constant function $\Phi:[0,\infty)\to[0,\infty]$ is called an {\em Orlicz function} if it is convex, non-decreasing, left continuous  and satisfies $\Phi(0)=0$.
Let $(\Omega,\cF,\bP)$ be a non-atomic probability space. For $X\in L^0$, its {\em Orlicz premium} (or {\em Luxemburg norm}) is defined by
\begin{align*}
 \|X\|_\Phi := \inf\left\{\lambda\in(0,\infty) \,: \ \E\left[\Phi\left(\frac{|X|}{\lambda}\right)\right]\leq1\right\}.
\end{align*}
The elements in $L^0$ with finite Luxemburg norm form the Orlicz space $L^\Phi$, namely,  
\[
L^\Phi := \left\{X\in L^0 \,: \ \|X\|_\Phi<\infty\right\}.
\]

\begin{example}
Let  $\alpha\in(0,1)$ be given. As is standard in the literature on  Haezendonck--Goovaerts risk measures, we assume that $\Phi$  is normalized by $\Phi(1)=1$. We consider the Orlicz function $\Phi_\alpha:=\frac{\Phi}{1-\alpha}$ and define the map $\rho^\alpha:L^0\to[0,\infty]$ by
\[
\rho^\alpha(X) := \|X^+\|_{\Phi_\alpha}.
\]
Let $\Pi$ be the constant one function $\mathbf{1}$. For $X\in L^0$, its {\em Haezendonck--Goovaerts risk measure} is defined by $$\pi_\alpha(X):=\inf_{m\in\R}\{m+\rho^\alpha(X-m\mathbf{1})\}=\rho_\mathbf{1}^\alpha(X).$$

It was proved by Gao, Munari, and Xanthos \cite{GMX20} that the Haezendonck--Goovaerts risk measure has the Fatou property on $L^\Phi$, and it was asked therein whether it satisfies the super Fatou property on $L^\Phi$.
The question was recently answered affirmatively by Amarante and Liebrich in \cite{AL24} using dual representations.   
We observe that it is an immediate consequence of Theorem~\ref{valuesofch}. 
    
To this end, we first recall that $\rho^\alpha$ is prudent on $L^0$. 
For  completeness, we include the proof. First, let $(X_n)\subset L^0$ and $X\in L^0$ be such that $0\leq X_n\uparrow X$ almost surely.   
    We claim that $\norm{X}_{\Phi_\alpha}=\sup_n\norm{X_n}_{\Phi_\alpha}$.  Since $\Phi$ is non-decreasing,  $\sup_n\norm{X_n}_{\Phi_\alpha}\leq \norm{X}_{\Phi_\alpha}$.
 Suppose, by way of contradiction, that $\sup_n\norm{X_n}_{\Phi_\alpha}<\norm{X}_{\Phi_\alpha}$.   
Then there exists a real number $M>0$ such that 
    $\norm{X_n}_{\Phi_\alpha}\leq M$ for all $n\in\N$ and $\norm{X}_{\Phi_\alpha}> M$. It follows that 
    \begin{align}\label{shortorlicz}
        \E\left[\Phi_\alpha\left(\frac{X_n}{M}\right)\right]\leq1\;\;\text{for all } n\geq 1,\qquad \E\left[\Phi_\alpha\left(\frac{X}{M}\right)\right]>1.
    \end{align}
Since $0\leq  \frac{X_n}{M}\uparrow \frac{X}{M}$ a.s., the left continuity and monotonicity of $\Phi_\alpha$ imply that $$0=\Phi_\alpha(0)\leq \Phi_\alpha\Big(\frac{X_n}{M}\Big)\uparrow \Phi_\alpha\Big(\frac{X}{M}\Big)\;\;\text{a.s.}$$ 
Hence, by the Monotone Convergence Theorem,
    $\E\left[\Phi_\alpha\left(\frac{X_n}{M}\right)\right]\uparrow \E\left[\Phi_\alpha\left(\frac{X}{M}\right)\right]$, which contradicts \eqref{shortorlicz}. Therefore, $\|X\|_{\Phi_\alpha}=\sup_n\|X_n\|_{\Phi_\alpha}$.    
Next, assume that $X_n\xrightarrow{a.s.} X$. Let $Y_n=\inf_{k\geq n}X_k^+$ for any $n\geq 1$. Then $0\leq Y_n\uparrow X^+$ a.s.  
By the preceding claim, $$\norm{X^+}_{\Phi_\alpha}=\sup_n\norm{Y_n}_{\Phi_\alpha}\leq \sup_n\inf_{k\geq n}\norm{X_k^+}_{\Phi_\alpha}=\liminf_n\norm{X_n^+}_{\Phi_\alpha},$$
showing that $\rho^\alpha$ is prudent on $L^0$.

Next, note that  the Orlicz premium $\norm{\cdot}_{\Phi_\alpha}$ is  positive homogeneous on $L^0$.  Thus, $\rho^\alpha$ is also positive homogeneous and hence star-shaped on $L^0$. 
Moreover, a straightforward computation yields $\rho^\alpha(0)=0$ so that $l_{\rho^\alpha}=0<1$. Similarly,  one obtains  $$u_{\rho^\alpha}=\norm{\one}_{\Phi_\alpha}=\inf\Big\{\lambda>0:\Phi\Big(\frac{1}{\lambda}\Big)\leq 1-\alpha\Big\}>1.$$
For the last inequality, first observe that $\Phi(x)\leq 1$ if and only if $x\leq 1$, since $\Phi(1)=1$ and $\Phi$ is convex and non-decreasing. The inequality then follows from the left  continuity and monotonicity of $\Phi$.     Therefore, by Theorem~\ref{valuesofch}, $\pi_\alpha$ is prudent on $L^0$ and, in particular, on $L^\Phi$.     
\end{example}

It is worth noting that the argument in the preceding example does not rely on the convexity of $\Phi$. 
Recently, Ayg\"un, Bellini, and  Laeven \cite{ABL25} introduced a broad framework of Orlicz premia, which in turn yields  a corresponding generalization of Haezendonck--Goovaerts risk measures.  Following \cite{ABL25}, let $\Phi: [0, \infty) \to (-\infty,\infty]$ be non-decreasing and left continuous, and assume  that  $x \leq 1$ if and only if  $\Phi(x) \leq 1$\footnote{The framework of \cite{ABL25} differs from the setting considered here in that it also allows $\Phi(0)=-\infty$.}.  Assume also that  $\Phi(1)=1$.
The Orlicz premium $\norm{\cdot}_{\Phi_\alpha}$, the functional $\rho^\alpha$ and the Haezendonck--Goovaerts risk measure $\pi_\alpha$ are defined exactly as before.  The same arguments as in the preceding example show that $\rho^\alpha$ is prudent and star-shaped on $L^0$ and satisfies $$l_{\rho^\alpha}=0<1,\qquad u_{\rho^\alpha}=\norm{\one}_{\Phi_\alpha}=\inf\Big\{\lambda>0:\Phi\Big(\frac{1}{\lambda}\Big)\leq 1-\alpha\Big\}>1.$$
Hence, Theorem~\ref{valuesofch} implies that $\pi_\alpha$ is prudent on $L^0$.

\section{Prudence of Inf-Convolutions}
In this section, we show that prudence is preserved under inf-convolution for convex, cash-additive, law-invariant functionals.
Let $\cX$ be an r.i.\ space over a non-atomic probability space $(\Omega,\cF,\bP)$.
For two proper functionals $\rho_i:\cX\rightarrow (-\infty,\infty]$, $i=1,2$, their inf-convolution $\rho_1\square\rho_2:\cX\to[-\infty,\infty]$ is defined by 
$$\rho_1\square\rho_2(X):=\inf\{\rho_1(Y)+\rho_2(Z):X=Y+Z,\;Y\in\cX,\;Z\in\cX\},\qquad \forall X\in \cX. $$

When $\cX=L^p$,  Filipovi\'c and Svindland  \cite{FS08} showed that if $\rho_1,\rho_2:\cX\to(-\infty,\infty]$ are both convex, cash additive, law invariant, and have the Fatou property, then $\rho_1\square\rho_2:\cX\to(-\infty,\infty]$ is convex, cash additive,
law invariant, and has the Fatou property. See also Acciaio \cite{A09}, Chen, Gao and Xanthos \cite{CGX18}, and  Liu, Wang, and Wei \cite{LWW20} for further results on preservation of properties by inf-convolutions. 

In particular, the result of Filipovi\'c and   Svindland  was extended in \cite{CGX18} from $L^p$-spaces to arbitrary r.i.\ spaces, with the Fatou property replaced by the strong Fatou property. It was later shown in \cite[Theorem 3.3]{ML20} that, on every r.i.\ space other than $L^1$, the Fatou property and the strong Fatou property coincide for convex, law-invariant functionals. Combined with the $L^1$ case in \cite{FS08}, this yields that inf-convolutions of convex, cash-additive, law-invariant functionals on arbitrary r.i.\ spaces  preserve the Fatou property and continue to take values in $(-\infty,\infty]$. We now show that prudence is also preserved under inf-convolution for such functionals.

\begin{theorem}\label{inf-con}
 Suppose that $\cX$ is an r.i.\ space. Let $\rho_i:\cX\rightarrow (-\infty,\infty]$ be proper, convex, cash-additive, law-invariant, prudent functionals, $i=1,2$. Then their inf-convolution $\rho_1\square\rho_2$ is also prudent.     
\end{theorem}

\begin{proof}
Let $(X_n)\subset \cX$ and $X\in\cX$ be such that $X_n\xrightarrow{a.s.}X$. We shall show  that  $\rho_1\square\rho_2(X)\leq \liminf_n\rho_1\square\rho_2(X_n)$. To this end, let   $m\in\R$  be  such that $\rho_1\square\rho_2(X_n)\leq m$ for all $n\in\N$. It suffices to prove that $\rho_1\square \rho_2(X)\leq m$.

Since prudence implies the Fatou property, by \cite[Theorem 2.5]{FS08} and \cite[Proposition 14]{CGX18}, for any $n\in\N$, there exist $f_n,g_n:\R\rightarrow \R$, both increasing, such that 
\begin{align}\label{inf-tech1}
f_n(x)+g_n(x)=x,\quad \forall x\in \R,
\end{align}
\begin{align*}
\rho_1\square\rho_2(X_n)=\rho_1(f_n(X_n))+\rho_2(g_n(X_n)).
\end{align*}
Since $\rho_i$'s are cash additive, by replacing $f_n$ and $g_n$ with $f_n-f_n(0)$ and $g_n-g_n(0)$, we may assume that 
\begin{align}\label{inf-tech3}
f_n(0)=g_n(0)=0.
\end{align}
Put $$Y_n=f_n(X_n),\quad Z_n=g_n(X_n).$$
By \eqref{inf-tech1}, $$X_n=Y_n+Z_n.$$
Since $f_n$'s and $g_n$'s are all increasing, by \eqref{inf-tech1} and \eqref{inf-tech3}, they all satisfy $\abs{f_n(x)}\leq \abs{x}$ and $\abs{g_n(x)}\leq \abs{x}$ for all $x\in\R$. 
It follows that 
\begin{align}\label{inf-tech4}
\abs{Y_n}\leq \abs{X_n},\quad \abs{Z_n}\leq \abs{X_n}.
\end{align}

Note that since $X_n\xrightarrow{a.s.}X$,  $\sup_{n\in\N}\abs{X_n}\in L^0$. Let $$\mathrm{d}\mu=\frac{1}{1+\sup_{n\in \N}{\abs{X_n}}}\mathrm{d}\bP.$$
Then $\mu$ is a finite measure on $(\Omega,\cF)$, equivalent to $\bP$.
Moreover, $\norm{Y_n}_{L^1(\mu)}\leq 1$ and $\norm{Z_n}_{L^1(\mu)}\leq 1$ for all $n\in\N$. Therefore, by applying Komlos' Theorem twice, we obtain subsequences $(Y_{n_k})$, $(Z_{n_k})$ and $Y,Z\in L^0$ such that 
$$\frac{1}{k}\sum_{j=1}^kY_{n_j}\xrightarrow{a.s.}Y,\qquad \frac{1}{k}\sum_{j=1}^kZ_{n_j}\xrightarrow{a.s.}Z.$$
Since $Y_{n_k}+Z_{n_k}=X_{n_k}\xrightarrow{a.s.}X$, we have $$Y+Z=X.$$
Moreover, by \eqref{inf-tech4}, $$\abs{Y}\leq \limsup_k\frac{1}{k}\sum_{j=1}^k\abs{Y_{n_j}}\leq\limsup_k\frac{1}{k}\sum_{j=1}^k\abs{X_{n_j}}=\abs{X}.$$
Since $\cX$ is an ideal of $L^0$, we obtain $Y\in \cX$. Similarly, $Z\in \cX$.

Using prudence and convexity of $\rho_i$'s, we obtain
\begin{align*}
\rho_1\square\rho_2(X)&\leq \rho_1(Y)+\rho_2(Z)\\
&\leq \liminf_k\rho_1\left(\frac{1}{k}\sum_{j=1}^k {Y_{n_j}}\right)+\liminf_k\rho_2\left(\frac{1}{k}\sum_{j=1}^k {Z_{n_j}}\right)\\
&\leq \liminf_k\frac{\sum_{j=1}^k {\rho_1(Y_{n_j}})}{k} + \liminf_k\frac{\sum_{j=1}^k {\rho_2(Z_{n_j}})}{k}\\
&\leq \liminf_k\frac{\sum_{j=1}^k \big(\rho_1(Y_{n_j})+\rho_2(Z_{n_j})\big)}{k} \\
&=\liminf_k\frac{\sum_{j=1}^k \rho_1\square \rho_2(X_{n_j})}{k} \\
&\leq m.
\end{align*}
This completes the proof.
\end{proof}

We conclude this section with two examples illustrating that the assumptions of prudence and cash additivity in Theorem~\ref{inf-con} must be imposed on both functionals in order for their inf-convolution to be prudent.

The first example shows that if only one of $\rho_1$ and $\rho_2$ is prudent, then their inf-convolution need not be prudent.

\begin{example}
Let $\cX$ be an r.i.\ space.  
Let $\rho_1:\cX\to(-\infty,\infty]$ be any proper, convex, cash-additive, law-invariant, and prudent functional.
Let $\rho_2:\cX\to \R$  be the expectation functional, that is,  $\rho_2(X) = \E[X]$ for all $X\in\cX$.  Then $\rho_2$ is convex, cash additive, law invariant, and satisfies the Fatou property, but it is not prudent.

We show that $ \rho_1\square\rho_2$ is not prudent at any $X\in \cX$ such that $\rho_1(X) <\infty$. Fix such an $X$. Choose a decreasing sequence of measurable sets $(A_n)_{n=1}^\infty$ with $\cap_{n=1}^\infty A_n= \emptyset$ and a sequence of real numbers $(t_n)_{n=1}^\infty$ such that  $$ t_n\bP(A_n)\to -\infty.$$
For each $n\geq 1$, define $$X_n = X + t_n\mathbf{1}_{A_n}.$$
Since $A_n\downarrow \emptyset$, we have $X_n\to X$ pointwise. Observe that
\[  \rho_1\square\rho_2(X_n) \leq\rho_1(X) + \rho_2(t_n\mathbf{1}_{A_n}) = \rho_1(X) +t_n\bP(A_n) \to -\infty.\]
Hence,  $\liminf_{n\to\infty} \rho_1\square\rho_2(X_n)=-\infty$. On the other hand, it follows from the discussion at the beginning of this section that  $\rho_1\square\rho_2(X)>-\infty$. Therefore, $\rho_1\square\rho_2(X)>\liminf_{n\to\infty} \rho_1\square\rho_2(X_n)$ and hence $\rho_1\square\rho_2$ is not prudent at $X$.
\end{example}

The next example shows that if only one of $\rho_1$ and $\rho_2$ is cash additive, then their inf-convolution may fail to be prudent.

\begin{example}
Let $\rho_1,\rho_2:L^\infty\to \R$ be  defined by
\[ \rho_1(X) = \E[X^+] \quad \text{and}\quad \rho_2(X) = \es X,\]
respectively.
Then both $\rho_1$ and $\rho_2$ are convex, law invariant, and prudent. Moreover, $\rho_2$ is cash additive, whereas $\rho_1$ is not. 

We show that $ \rho_1\square \rho_2$ is not prudent at any $X\in L^\infty$.
To this end, it suffices to prove that  $$\rho_1\square \rho_2(X) = \E[X],\qquad X\in L^\infty.$$
Fix $X\in L^\infty$. Choose $c\in\R$  such that $X\geq c $ a.s.
Then 
\[\rho_1\square \rho_2(X) \leq \rho_1(X-c\mathbf{1}) + \rho_2(c\mathbf{1}) = \E[X-c\mathbf{1}] + c = \E[X].\]
Conversely, let $Y\in L^\infty$ and set $c = \es Y$.
Since $Y\leq c$ a.s.,
\[
\E[(X-Y)^+] \geq \E[(X-c\mathbf{1})^+] \geq \E[X-c\mathbf{1}] = \E[X] -c.\]
Hence,
\[ \rho_1(X-Y) + \rho_2(Y) \geq \E[X] - c + c = \E[X].\]
Taking the infimum over $Y\in L^\infty$, we obtain  that $$\rho_1\square \rho_2(X) \geq \E[X].$$
Combining the two parts gives $\rho_1\square \rho_2(X) = \E[X]$. 
\end{example}


{\footnotesize

\begin{thebibliography}{99}
 
  \bibitem{A09}
Acciaio, B.: Short note on inf-convolution preserving the Fatou property, \emph{Annals of Finance} 5(2),  281--287 (2009).

\bibitem{AS14}
Ahn, J.Y., Shyamalkumar, N.D.: Asymptotic theory for the empirical Haezendonck--Goovaerts risk measure,
\emph{Insurance: Mathematics and Economics} 55, 78--90 (2014).

\bibitem{AL24}
Amarante, M., Liebrich, F.B.: Distortion risk measures: Prudence, coherence, and the expected shortfall,
\emph{Mathematical Finance} 34(4), 1291--1327 (2024).

\bibitem{ABL25}
Ayg\"un, M., Bellini, F., Laeven, R.J.: Generalized Orlicz premia,
arXiv:2507.09181 (2025).

\bibitem{BFWW22}Bellini, F., Fadina, T., Wang, R., Wei, Y.:  Parametric measures of variability induced by risk measures, \emph{Insurance, Mathematics and Economics} 106, 270-284 (2022)

\bibitem{BKMS21}
Bellini, F., Koch-Medina, P., Munari, C., Svindland, G.: Law-Invariant Functionals on General Spaces of Random Variables,
\emph{SIAM Journal on Financial Mathematics} 12(1), 318–341 (2021).

\bibitem{BR12}
Bellini, F., Rosazza Gianin, E.: Haezendonck--Goovaerts risk measures and Orlicz quantiles, {\em Insurance: Mathematics
and Economics} 51, 107-114 (2012).

\bibitem{CCMTW22}
Castagnoli, E., Cattelan, G., Maccheroni, F., Tebaldi, C., Wang, R.:
Star-shaped risk measures,
\emph{Operations Research} 70(5), 2637--2654 (2022).

\bibitem{CGLL22}
Chen, S., Gao, N., Leung, D., Li, L.:
Automatic Fatou property of law-invariant risk measures,
\emph{Insurance: Mathematics and Economics} 105, 41--53 (2022).

\bibitem{CGX18}
Chen, S., Gao,  N., Xanthos, F.: The strong Fatou property of risk measures, {\em Dependence Modeling} 6(1), 183-196 (2018).

\bibitem{D02}Delbaen, F.: Coherent risk measures on general probability spaces, In: \emph{Advances in finance and stochastics}, Springer Berlin Heidelberg, 2002, 1–37.

\bibitem{DPR10}
Dentcheva, D., Penev, S., Ruszczy\'nski, A.:
Kusuoka representation of higher order dual risk measures,
\emph{Annals of Operations Research} 181(1), 325--335 (2010).

\bibitem{FK07}
Filipovi\'c, D., Kupper, M.:
Monotone and cash-invariant convex functions and hulls,
\emph{Insurance: Mathematics and Economics} 41(1), 1--6 (2007).

\bibitem{FS08}
Filipovi\'c, D., Svindland, G.:
Optimal capital and risk allocations for law- and cash-invariant convex functions,
\emph{Finance and Stochastics} 12, 423--439 (2008).


\bibitem{GM20}
Gao, N., Munari, C.:
Surplus-invariant risk measures,
\emph{Mathematics of Operations Research} 45(4), 1342--1370 (2020).

\bibitem{GMX20}
Gao, N., Munari, C., Xanthos, F.:
Stability properties of Haezendonck--Goovaerts premium principles,
\emph{Insurance: Mathematics and Economics} 94, 94--99 (2020).


\bibitem{GLMX18}
Gao, N., Leung, D., Munari, C., Xanthos, F.:
Fatou property, representations, and extensions of law-invariant risk measures on general Orlicz spaces,
\emph{Finance and Stochastics} 22(2), 395--415 (2018).

\bibitem{GLX19}
Gao, N., Leung, D., Xanthos, F.:
Closedness of convex sets in Orlicz spaces with applications to dual representation of risk measures,
\emph{Studia Mathematica} 249, 329--347 (2019).

\bibitem{GX18}
Gao, N., Xanthos, F.: On the C-property and $w^*$-representations of risk measures, \emph{Mathematical Finance} 28(2), 748–754 (2018).

\bibitem{GX25}
Gao, N., Xanthos, F.:
A note on continuity and asymptotic consistency of measures of risk and variability,
\emph{ASTIN Bulletin: The Journal of the IAA} 55(1), 168--177 (2025).

\bibitem{JST06}
Jouini, E., Schachermayer, W., Touzi, N.: Law invariant risk measures have the Fatou Property, \emph{Advances in Mathematical Economics} 9, 49–71 (2006).

\bibitem{K07}
Krohmal, P.:
Higher moment coherent risk measures,
\emph{Quantitative Finance} 7(4), 373--387 (2007).

\bibitem{LGZ23}
Laeven, R.J., Rosazza Gianin, E., Zullino, M.:
Dynamic return and star-shaped risk measures via BSDEs,
arXiv:2307.03447 (2023).

\bibitem{LGZ24}
Laeven, R.J., Rosazza Gianin, E., Zullino, M.:
Law-invariant return and star-shaped risk measures,
\emph{Insurance: Mathematics and Economics} 117, 140--153 (2024).

\bibitem{L24}
Liebrich, F.:
Risk sharing under heterogeneous beliefs without convexity,
\emph{Finance and Stochastics} 28, 999--1033 (2024).

\bibitem{LM25}
Liebrich, F., Munari, C.: Revisiting the Automatic Fatou Property of Law-Invariant Functionals, \emph{SIAM Journal on Financial Mathematics} 16(1), 2025.

\bibitem{LWW20}
Liu, P., Wang, R., Wei, L.:
Is the inf-convolution of law-invariant preferences law-invariant?
\emph{Insurance: Mathematics and Economics} 91, 144--154 (2020).

\bibitem{BTW23}
O'Brien, M., Troitsky, V.G., van der Walt, J.H.:
Net convergence structures with applications to vector lattices,
\emph{Quaestiones Mathematicae} 46(2), 243--280 (2023).

\bibitem{RX20}
Rahsepar, M., Xanthos, F.:
On the extension property of dilatation monotone risk measures,
\emph{Statistics \& Risk Modeling} 37, 139--152 (2020).

\bibitem{R24}
Righi, M.B.:
Star-shaped acceptability indexes,
\emph{Insurance: Mathematics and Economics} 117, 170--181 (2024).

\bibitem{ML20}
Tantrawan, M., Leung, D.H.: On closedness of law-invariant convex sets in rearrangement invariant spaces. \emph{Arch. Math.} 114, 175–183 (2020).


\bibitem{WZ21}
Wang, R., Zitikis, R.:
An axiomatic foundation for the expected shortfall,
\emph{Management Science} 67(3), 1413--1429 (2021).

\bibitem{WMH24}
Wu, Q., Mao, T., Hu, T.:
Generalized optimized certainty equivalent with applications in the rank-dependent utility model,
\emph{SIAM Journal on Financial Mathematics} 15(1), 255--294 (2024).

 
\end{thebibliography}

}
\end{document}